\documentclass{llncs}
\usepackage[utf8]{inputenc}
\usepackage{graphicx}
\usepackage{amsmath}
\usepackage{amssymb}
\usepackage{amsfonts}

\newcommand{\ind}[1]{\textbf{1}_{\left \{ #1 \right \}}}
\newtheorem{propn}{Proposition}

\DeclareMathOperator*{\argmin}{\arg\!\min}

\title{Anti-Coordination Games and Stable \\Graph Colorings}

\begin{document}
\author{Jeremy Kun \and Brian Powers \and Lev Reyzin}

\institute{Department of Mathematics, Statistics, and Computer Science\\
University of Illinois at Chicago\\
\texttt{\{jkun2,bpower6,lreyzin\}@math.uic.edu}}

\maketitle

\begin{abstract}
Motivated by understanding non-strict and strict pure strategy equilibria in
network anti-coordination games, we define notions of stable and, respectively,
strictly stable colorings in graphs.  We characterize the cases when such
colorings exist and when the decision problem is NP-hard. These correspond to
finding pure strategy equilibria in the anti-coordination games, whose price of
anarchy we also analyze.  We further consider the directed case, a
generalization that captures both coordination and anti-coordination. We prove
the decision problem for non-strict equilibria in directed graphs is NP-hard.
Our notions also have multiple connections to other combinatorial questions, and
our results resolve some open problems in these areas, most notably the
complexity of the strictly unfriendly partition problem.
\end{abstract}

\section{Introduction}

Anti-coordination games form some of the basic payoff structures in game theory.
Such games are ubiquitous; miners deciding which land to drill for resources,
company employees trying to learn diverse skills, and airplanes selecting flight
paths all need to mutually anti-coordinate their strategies in order to maximize
their profits or even avoid catastrophe.

Two-player anti-coordination is simple and well understood.  In its barest form,
the players have two actions, and payoffs are symmetric for the players, paying
off $1$ if the players choose different actions and $0$ otherwise.  This game
has two strict pure-strategy equilibria, paying off $1$ to each player, as well
as a non-strict mixed-strategy equilibrium paying off an expected $1/2$ to each
player.

In the real world, however, coordination and anti-coordination games are more
complex than the simple two-player game.  People, companies, and even countries
play such multi-party games simultaneously with one another. One straightforward
way to model this is with a graph, whose vertices correspond to agents and whose
edges capture their pairwise interactions.  A vertex then chooses one of $k$
strategies, trying to anti-coordinate with all its neighbors simultaneously.
The payoff of a vertex is the sum of the payoffs of its games with its neighbors
-- namely the number of neighbors with which it has successfully
anti-coordinated.  It is easy to see that this model naturally captures many
applications.  For example countries may choose commodities to produce, and
their value will depend on how many trading partners do not produce that
commodity.

In this paper we focus on finding \text{pure strategies} in equilibrium, as well
as their associated social welfare and price of anarchy, concepts we shall
presently define.  We look at both strict and non-strict  pure strategy
equilibria, as well as games on directed and undirected graphs.  Directed graphs
characterize the case where only one of the vertices is trying to
anti-coordinate with another.  The directed case turns out to not only
generalize the symmetric undirected case, but also captures coordination in
addition to anti-coordination.

These problems also have nice interpretations as certain natural graph coloring
and partition problems, variants of which have been extensively studied.  For
instance, a pure strategy equilibrium in an undirected graph corresponds to what
we call a stable $k$-coloring of the graph, in which no vertex can have fewer
neighbors of any color different than its own.  For $k=2$ colors this is
equivalent to the well-studied \emph{unfriendly partition} or
\emph{co-satisfactory partition} problem.  The strict equilibrium version of
this problem (which corresponds to what we call a strictly stable $k$-coloring)
generalizes the \emph{strictly unfriendly partition problem}. We establish both
the NP-hardness of the decision problem for strictly unfriendly partitions and
NP-hardness for higher $k$.

\subsection{Previous work}
In an early work on what can be seen as a coloring game, Naor and
Stockmeyer~\cite{NaorS93} define a \emph{weak $k$-coloring} of a graph to be one
in which each vertex has a differently colored neighbor.  They give a locally
distributed algorithm that, under certain conditions, weakly $2$-colors a graph
in constant time. 

Then, in an influential experimental study of anti-coordination in networks,
Kearns~et~al.~\cite{KearnsSM06} propose a true graph coloring game, in which
each participant controlled the color of a vertex, with the goal of coloring a
graph in a distributed fashion.  The players receive a reward only when a proper
coloring of the graph is found.  The theoretical properties of this game are
further studied by Chaudhuri~et~al.~\cite{ChaudhuriGJ08} who prove that in a
graph of maximum degree $d$, if players have $d + 2$ colors available they will
w.h.p.\ converge to a proper coloring rapidly using a greedy local algorithm.
Our work is also largely motivated by the work of Kearns~et~al., but for a
somewhat relaxed version of proper coloring.

Bramoull\'{e}~et~al.~\cite{BramoulleLGV04} also study a general
anti-coordination game played on networks.  In their formulation, vertices can
choose to form links, and the payoffs of two anti-coordinated strategies may not
be identical.  They go on to characterize the strict equilibria of such games,
as well as the effect of network structure on the behavior of individual agents.
We, on the other hand, consider an arbitrary number of strategies but with a
simpler payoff structure.

The game we study is related to the MAX-$k$-CUT game, in which each player
(vertex) chooses its place in a partition so as to maximize the number of
neighbors in other partitions. Hoefer~\cite{Hoefer2007}, Monnot \&
Gourv\`es~\cite{G09}, research Nash equlibria and coalitions in this context.
Our Propositions~\ref{propn:alg} and~\ref{obs:poa} generalize known facts proved
there, and we include them for completeness.

This paper also has a strong relationship to \emph{unfriendly partitions} in
graph theory.  An unfriendly partition of a graph is one in which each vertex
has at least as many neighbors in other partitions as in its own.  This topic
has been extensively studied, especially in the combinatorics
community~\cite{AharoniMP90,BruhnDGS10,CowanE,ShelahM90}.  While locally finite
graphs admit $2$-unfriendly partitions, uncountable graphs may
not~\cite{ShelahM90}.

Friendly (the natural counterpart) and unfriendly partitions are also studied
under the names \emph{max satisfactory} and \emph{min co-satisfactory
partitions} by Bazgan~et~al.~\cite{BazganTV10}, who focus on partitions of size
greater than $2$.  They characterize the complexity of determining whether a
graph has a $k$-friendly partition and asked about characterizing $k$-unfriendly
partitions for $k > 2$.  Our notion of stable colorings captures unfriendly
partitions, and we also solve the $k>2$ case.

A natural strengthening of the notion above yields \emph{strictly unfriendly
partitions},  defined by Shafique and Dutton~\cite{ShafiqueD09}.  A strictly
unfriendly partition requires each vertex to have strictly more neighbors
outside its partition than inside it.  Shafique and Dutton characterize a weaker
notion, called \emph{alliance-free partition}, but leave characterizing strictly
unfriendly partitions open.  Our notion of strictly stable coloring captures
strictly unfriendly partitions, giving some of the first results on this
problem.  Cao and Yang~\cite{CaoY12a} also study a related problem originating
from sociology, called the \emph{matching pennies game}, where some vertices try
to coordinate and others try to anti-coordinate.  They prove that deciding
whether such a game has a pure strategy equilibrium is NP-Hard.  Our work on the
directed case generalizes their notion (which they suggested for future work).
Among our results we give a simpler proof of their hardness result for $k=2$ and
also tackle  $k >2$, settling one of their open questions.

There are a few related games on graphs that involve coloring, but they instead
focus on finding good proper colorings. In~\cite{PS08} Panagopoulou and Spirakis
define a coloring game in which the payoff for a vertex is either zero if it
shares a color with a neighbor, and otherwise the number of vertices in the
graph with which it shares a color. They prove pure Nash equilibria always exist
and can be efficiently computed, and provide nice bounds on the number of colors
used. Chatzigiannakis, et al.~\cite{CKPS10} extend this line of work by
analyzing distributed algorithms for this game, and Escoffier, et
al.~\cite{EGM12} improve their bounds.

\subsection{Results}

We provide proofs of the following, the last two being our main results.

\begin{enumerate}

\item \emph{For all $k \ge 2$, every undirected graph has a stable
$k$-coloring, and such a coloring can be found in polynomial time.}  \\ Our
notion of stable $k$-colorings is a strengthening of the notion of
$k$-unfriendly partitions of Bazgan~et~al.~\cite{BazganTV10}, solving their
open problem number 15.

\item \emph{For undirected graphs, the price of anarchy for stable $k$-colorings is
bounded by $\frac{k}{k-1}$, and this bound is tight.}

\item \emph{In undirected graphs, for all $k \ge 2$, determining whether a
graph has a strictly stable $k$-coloring is NP-hard.}  \\ For $k=2$, this
notion is equivalent to the notion that is defined by Shafique and
Dutton~\cite{ShafiqueD09}, but left unsolved.

\item \emph{For all $k \ge 2$, determining whether a directed graph has even a
non-strictly stable $k$-coloring is NP-hard.}\\  Because directed graphs also
capture coordination, this solves two open problems of Cao and
Yang~\cite{CaoY12a}, namely generalizing the coin matching game to more than
two strategies and considering the directed case. 

\end{enumerate}

\section{Preliminaries}

For an unweighted undirected graph $G=(V,E)$, let $C = \{f | f: V
\to \{1, \ldots ,k \}\}.$ We call a function $c \in C$ a \textbf{coloring}.

We study the following anti-coordination game played on a graph $G=(V,E)$.  In
the game, all vertices simultaneously choose a color, which induces a coloring
$c \in C$ of the graph.  In a given coloring $c$, an agent $v$'s
\textbf{payoff}, $\mu_c(v)$, is the number of neighbors choosing colors
different from $v$'s, namely 
\[ 
   \mu_c(v) := \sum_{\{v,w\} \in E} \ind{c(v) \neq c(w)}.  
\] 
Note that in this game higher degree vertices have higher potential payoffs.

We also have a natural generalization to directed graphs.  That is, if $G =
(V,E)$ is a directed graph and $c$ is a coloring of $V$, we can define the
payoff $\mu_c(v)$ of a vertex $v \in V$ analogously as the sum over outgoing
edges:
\[ 
	\mu_c(v) := \sum_{(v,w) \in E} \ind{c(v) \neq c(w)}
\]
Here a directed edge from $v$ to $w$ is interpreted as ``$v$ cares about $w$.''
We can then define the social welfare and price of anarchy for directed graphs
identically using this payoff function. 

Given a graph $G$, we define the \textbf{social welfare} of a coloring $c$ to
be
\[
   W(G,c) := \sum_{v \in V} \mu_c(v).
\]
We say a coloring $c$ is \textbf{stable}, or in {equilibrium}, if no vertex can
improve its payoff by changing its color from $c(v)$ to another color. We define
$Q$ to be the set of stable colorings.

We call a coloring function $c$ \textbf{strictly stable}, or in {strict
equilibrium}, if every vertex would decrease its payoff by changing its color
from $c(v)$ to another color. If a coloring function is stable and at least one
vertex can change its color without decreasing its payoff, then the coloring is
\textbf{non-strict}.

We define the \textbf{price of anarchy} for a graph $G$ to be
\[
\mbox{PoA}(G) := \frac{\max_{c' \in C}W(G,c')}
{\min_{c \in Q}W(G,c)}.
\]
This concept was originally introduced by Koutsoupias and Papadimitriou
in~\cite{KP99}, where they consider the ratio of social payoffs in the best and
worst-case Nash equilibria. Much work has since focused on the price of anarchy,
e.g.~\cite{FKKMS02,RT02}.\\

\noindent \textbf{Mixed and pure strategies}\ \
It is natural to consider both pure and mixed strategies for the players in our
network anti-coordination game.  A pure strategy solution does not in general
exist for every 2 player game, while a mixed strategy solution will.  However,
in this coloring game not only will a pure strategy solution always exist, but
for any mixed strategy solution there is a pure strategy equilibrium solution
which achieves a social welfare at least as good, and where each player's payoff
is identical with its expected payoff under the mixed strategy.\\

\noindent \textbf{Strict and non-strict stability}\ \
It is worthwhile to note that a strictly stable coloring $c$ need not provide
the maximum social welfare.  In fact, it is not difficult to construct a graph
for which a strictly stable coloring exists yet the maximum social welfare is
achieved by a non-strictly stable coloring, as shown in
Figure~\ref{fig:weakstrongwelfare}. 
\begin{figure}[t]
\centering
\scalebox{0.35}{\includegraphics{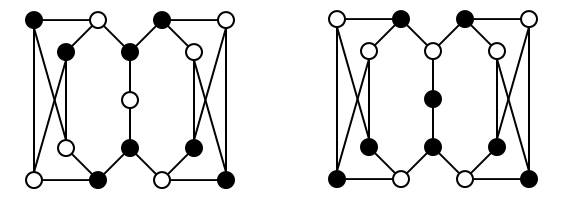}}
\caption{The strictly stable 2-coloring on the left attains a social welfare of
40 while the non-strictly stable coloring on the right attains
42, the maximum for this graph.}
\label{fig:weakstrongwelfare}
\end{figure}

\section{Stable colorings}

First we consider the problem of finding stable colorings in graphs.  For the
case $k=2$, this is equivalent to the solved unfriendly partition problem.  For
this case our algorithm is  equivalent to the well-studied local algorithm for
MAX-CUT~\cite{ElsasserT11,MonienT10}.  Our argument is a variant of a standard
approximation algorithm for MAX-CUT, generalized to work with partitions of size
$k \ge 2$.

\begin{propn}\label{propn:alg}
For all $k \ge 2,$ every finite graph $G=(V,E)$ admits a stable $k$-coloring.
Moreover, a stable $k$-coloring can be found in polynomial time.
\end{propn}

\begin{proof}
Given a coloring $c$ of a graph, define $\Phi(c)$ to be the number of
properly-colored edges. It is clear that this function is bounded and that
social welfare is $2 \Phi(c)$. Moreover, the change in a vertex's utility by
switching colors is exactly the change in $\Phi$, realizing this as an exact
potential game~\cite{M96}. In a given coloring, we call a vertex $v$
\emph{unhappy} if $v$ has more neighbors of its color than of some other color.
We now run the following process: while any unhappy vertex exists, change its
color to the color
\begin{equation}\label{eq:greedy}
c'(u) = \argmin_{m \in \{1, \ldots, k\}} \sum_{ v \in N(u)}\ind{c(v) = m}.
\end{equation}
As we only modify the colors of unhappy vertices, such an amendment to a
coloring increases the value of $\Phi$ by at least 1. After at most $|E|$ such
modifications, no vertex will be unhappy, which by definition means the coloring
is stable. \hfill $\square$
\end{proof}

We note that because, in the case of $k=2$, maximizing the social welfare of a
stable coloring is equivalent to finding the MAX-CUT of the same graph, which is
known to be NP-hard~\cite{GareyJ79}, we cannot hope to find a global optimum for
the potential function.  However, we can ask about the price of anarchy, for
which we obtain a tight bound.  The following result also appears, using a
different construction, in~\cite{Hoefer2007}, but we include it herein for
completeness.

\begin{propn}\label{obs:poa}
The price of anarchy of the $k$-coloring anti-coordination game is at most
$\frac{k}{k-1}$, and this bound is tight.
\end{propn}

\begin{proof}
By the pigeonhole principle, each vertex can always achieve a $\frac{k-1}{k}$
fraction of its maximum payoff by choosing its color according to
Equation~\ref{eq:greedy}.  Hence, if some vertex does not achieve this payoff
then the coloring is not stable.  This implies that the price of anarchy is at
most $\frac{k}{k-1}$.

To see that this bound is tight take two copies of $K_k$ on vertices $v_1,
\dots, v_k$ and $v_{k+1}, \dots, v_{2k}$ respectively. Add an edge joining $v_i$
with $v_{i+k}$ for $i\in \{1,\dots,k\}$. If each vertex $v_i$ and $v_{i+k}$ is
given color $i$ this gives a stable $k$-coloring of the graph, as each vertex
has one neighbor of each of the $k$ colors attaining the social welfare lower
bound of $2(\frac{k-1}{k})|E|$. If, however, the vertices $v_{i+k}$ take color
$i+1$ for $i\in\{1,\dots,k-1\}$ and $v_{2k}$ takes color 1, the graph achieves
the maximum social welfare of $2|E|$.  This is illustrated for $k=5$ in
Figure~\ref{fig:k5copies}.
\hfill
$\square$\end{proof}

\begin{figure}[htb]
\centering
\scalebox{0.40}{\includegraphics{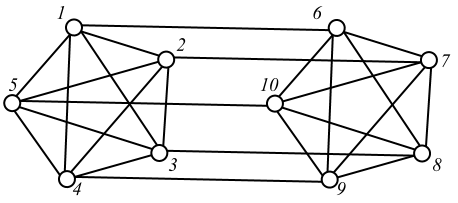}}
\caption{A graph achieving PoA of $\frac{5}{4}$, for k=5}
\label{fig:k5copies}
\end{figure}

\section{Strictly Stable Colorings}

In this section we show that the problem of finding a strictly stable
equilibrium with any fixed number $k \geq 2$ of colors is NP-complete.  We give
NP-hardness reductions first for $k \geq 3$ and then for $k=2$.  The $k=2$ case
is equivalent to the strictly unfriendly $2$-partition
problem~\cite{ShafiqueD09}, whose complexity we settle.

\begin{theorem} 
For all $k \geq 2$, determining whether a graph has a strictly stable
$k$-coloring is NP-complete.  
\end{theorem}

\begin{proof}
This problem is clearly in NP.  We now analyze the hardness in two cases.

\noindent \emph{1)} $k\ge3$:
For this case we reduce from classical $k$-coloring.  Given a graph $G$, we
produce a graph $G'$ as follows.

Start with $G' = G$, and then for each edge $e = \{ u,v \}$ in $G$ add a copy
$H_e$ of $K_{k-2}$ to $G'$ and enough edges s.t.\ the
induced subgraph of $G'$ on $V(H_e) \cup \left \{ u,v \right \}$
is the complete graph on $k$ vertices. Figure~\ref{fig:edgegadget} illustrates
this construction.

\begin{figure}[htb]
\centering
\scalebox{0.4}{\includegraphics{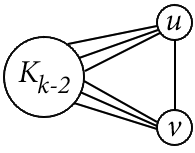}}
\caption{The gadget added for each edge in $G$.}
\label{fig:edgegadget}
\end{figure}

Now supposing that $G$ is $k$-colorable, we construct a strictly stable
equilibrium in $G'$ as follows. Fix any proper $k$-coloring $\varphi$ of $G$.
Color each vertex in $G'$ which came from $G$ (which is not in any $H_e$) using
$\varphi$.  For each edge $e = (u,v)$ we can trivially assign the remaining
$k-2$ colors among the vertices of $H_e$ to put the corresponding copy of $K_k$
in a strict equilibrium. Doing this for every such edge results in a strictly
stable coloring. Indeed, this is a proper $k$-coloring of $G'$ in which every
vertex is adjacent to vertices of all other $k-1$ colors.

Conversely, suppose $G'$ has a strictly stable equilibrium with $k$ colors.
Then no edge $e$ originally coming from $G$ can be monochromatic. If it were,
then there would be $k-1$ remaining colors to assign among the remaining $k-2$
vertices of $H_e$. No matter the choice, some color is unused and any vertex of
$H_e$ could change its color without penalty, contradicting that $G'$ is in a
strict equilibrium.

The only issue is if $G$ originally has an isolated vertex. In this case, $G'$
would have an isolated vertex, and hence will not have a strict equilibrium
because the isolated vertex may switch colors arbitrarily without decreasing its
payoff.  In this case, augment the reduction to attach a copy of $K_{k-1}$ to
the isolated vertex, and the proof remains the same.

\noindent \emph{2)} $k =2$:
We reduce from 3-SAT. Let $\varphi
= C_1 \wedge \dots \wedge C_k$ be a boolean formula in 3-CNF form. We construct
a graph $G$ by piecing together gadgets as follows.

For each clause $C_i$ construct an isomorphic copy of the graph shown in
Figure~\ref{fig:clausegadget}. We call this the \emph{clause gadget} for $C_i$.
In Figure~\ref{fig:clausegadget}, we label certain vertices to show how the
construction corresponds to a clause.  We call the two vertices labeled by the
same literal in a clause gadget a \emph{literal gadget.} In particular,
Figure~\ref{fig:clausegadget} would correspond to the clause $(x \vee y \vee
\bar{z})$, and a literal assumes a value of true when the literal gadget is
monochromatic. Later in the proof we will force literals to be consistent across
all clause gadgets, but presently we focus on the following key property of a
clause gadget.

\begin{figure}[t]
\centering
\scalebox{0.37}{\includegraphics{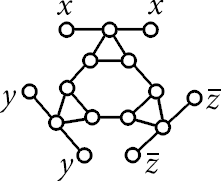}}
\caption{The clause gadget for $(x \vee y \vee \bar{z})$. Each literal
corresponds to a pair of vertices, and a literal being satisfied corresponds
to both vertices having the same color.}
\label{fig:clausegadget}
\end{figure}

\begin{lemma}
\label{lemma:clausegadget}
Any strictly stable 2-coloring of a clause gadget has a monochromatic literal
gadget. Moreover, any coloring of the literal gadgets which includes a
monochromatic literal extends to a strictly stable coloring of the clause
gadget (excluding the literal gadgets).
\end{lemma}

\begin{proof}
The parenthetical note will be resolved later by the high-degree of the vertices
in the literal gadgets. Up to symmetries of the clause gadget (as a graph) and
up to swapping colors, the proof of Lemma~\ref{lemma:clausegadget} is
illustrated in Figure~\ref{fig:clauselemmaproof}. The first five graphs show the
cases where one or more literal gadgets are monochromatic, and the sixth shows
how no strict equilibrium can exist otherwise. Using the labels in
Figure~\ref{fig:clauselemmaproof}, whatever the choice of color for the vertex
$v_1$, its two uncolored neighbors must have the same color (or else $v_1$ is
not in strict equilibrium). Call this color $a$. For $v_2, v_3$, use the same
argument and call the corresponding colors $b, c$, respectively. Since there are
only two colors, one pair of $a,b,c$ must agree.  WLOG suppose $a=b$. But then
the two vertices labeled by $a$ and $b$ which are adjacent are not in strict
equilibrium.  \hfill $\square$ 
\end{proof}

\begin{figure}[h]
\centering
\scalebox{0.5}{\includegraphics{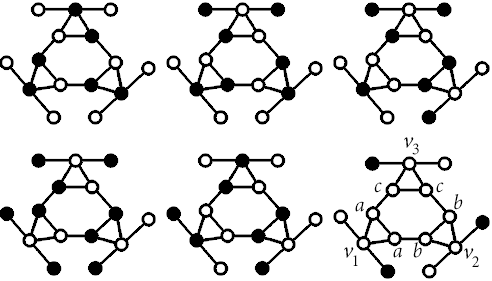}}
\caption{The first five figures show
that a coloring with a monochromatic literal gadget can be extended to a strict
equilibrium. The sixth (bottom right) shows that no strict equilibrium can
exist if all the literals are not monochromatic.}
\label{fig:clauselemmaproof}
\end{figure}

Using Lemma~\ref{lemma:clausegadget}, we complete the proof of the theorem. We
must enforce that any two identical literal gadgets in different clause gadgets
agree (they are both monochromatic or both not monochromatic), and that any
negated literals disagree. We introduce two more simple gadgets for each
purpose.

The first is for literals which must agree across two clause gadgets, and we
call this the \emph{literal persistence gadget}. It is shown in
Figure~\ref{fig:connectiongadgets}. The choice of colors for the literals on one
side determines the choice of colors on the other, provided the coloring is
strictly stable. In particular, this follows from the central connecting vertex
having degree 2. A nearly identical argument applies to the second gadget, which
forces negated literals to assume opposite truth values. We call this the
\emph{literal negation gadget}, and it is shown in
Figure~\ref{fig:connectiongadgets}.  We do not connect all matching literals
pairwise by such gadgets but rather choose one reference literal $x'$ per
variable and connect all literals for $x, \overline{x}$ to $x'$ by the needed
gadget.

\begin{figure}[htb]
\centering
\scalebox{0.45}{\includegraphics{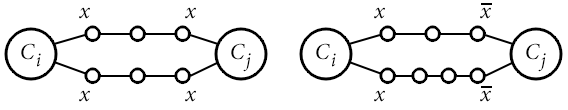}}
\caption{The literal persistence gadget (left) and literal negation gadget
(right) connecting two clause gadgets $C_i$ and $C_j$. The vertices labeled $x$
on the left are part of the clause gadget for $C_i$, and the vertices labeled
$x$ on the right are in the gadget for $C_j$.}
\label{fig:connectiongadgets}
\end{figure}

The reduction is proved in a straightforward way. If $\varphi$ is satisfiable,
then monochromatically color all satisfied literal gadgets in $G$. We can extend
this to a stable 2-coloring: all connection gadgets and unsatisfied literal
gadgets are forced, and by Lemma~\ref{lemma:clausegadget} each clause gadget can
be extended to an equilibrium. By attaching two additional single-degree
vertices to each vertex in a literal gadget, we can ensure that the literal
gadgets themselves are in strict equilibrium and this does not affect any of the
forcing arguments in the rest of the construction.

Conversely, if $G$ has a strictly stable 2-coloring, then each clause gadget has
a monochromatic literal gadget which gives a satisfying assignment of $\varphi$.
All of the gadgets have a constant number of vertices so the construction is
polynomial in the size of $\varphi$. This completes the reduction and proves the
theorem.  \hfill $\square$ 
\end{proof}

\section{Stable colorings in directed graphs}

In this section we turn to directed graphs.  The directed case clearly
generalizes the undirected as each undirected edge can be replaced by two
directed edges.  Moreover, directed graphs can capture coordination. For two
colors, if vertex $u$ wants to coordinate with vertex $v$, then instead of
adding an edge $(u,v)$ we can add a proxy vertex $u'$ and edges $(u,u')$ and
$(u',v)$. To be in equilibrium, the proxy has no choice but to disagree with
$v$, and so $u$ will be more inclined to agree with $v$. For $k$ colors we can
achieve the same effect by adding an undirected copy of $K_{k-1}$, appropriately
orienting the edges, and adding edges $(u,x), (x,v)$ for each $x \in K_{k-1}$.
Hence, this model is quite general.

Unlike in the undirected graph case, a vertex updating its color according to
Equation~\ref{eq:greedy} does not necessarily improve the overall social
welfare. In fact, we cannot guarantee that a pure strategy equilibrium even
exists -- e.g.\ a directed $3$-cycle has no stable 2-coloring, a fact that we
will use in this section.

We now turn to the problem of determining if a directed graph has an equilibrium
with $k$ colors and prove it is NP-hard.  Indeed, for strictly stable colorings
the answer is immediate by reduction from the undirected case. Interestingly
enough, it is also NP-hard for non-strict $k$-colorings for any $k \geq 2$. 

\begin{theorem} 
For all $k \geq 2$, determining whether a directed graph has a
stable $k$-coloring is NP-complete.
\end{theorem}

\begin{proof}
This problem is clearly in NP.  We again separate the hardness analysis into two
parts: $k=2$ and $k \geq 3$. 

\noindent \emph{1)} $k=2$:
We reduce from the balanced unfriendly partition problem. A balanced 2-partition
of an undirected graph is called unfriendly if each vertex has at least as many
neighbors outside its part as within.  Bazgan et al. proved that the decision
problem for balanced unfriendly partitions is NP-complete~\cite{BazganTV10}.
Given an undirected graph $G$ as an instance of balanced unfriendly partition,
we construct a directed graph $G'$ as follows.

Start by giving $G'$ the same vertex set as $G$, and replace each undirected
edge of $G$ with a pair of directed edges in $G'$. Add two vertices $u,v$ to
$G'$, each with edges to the other and to all other vertices in $G'$. Add an
additional vertex $w$ with an edge $(w,v)$, and connect one vertex of a directed
3-cycle to $u$ and to $w$, as shown in Figure~\ref{fig:weaktwocolornphard}.

\begin{figure}[htb]
\centering
\scalebox{0.4}{\includegraphics{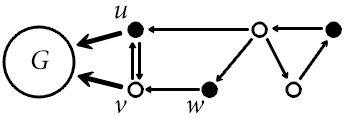}}
\caption{The construction from balanced unfriendly partition to directed stable
2-coloring. Here $u$ and $v$ ``stabilize'' the 3-cycle. A bold arrow denotes a
complete incidence from the source to the target.}
\label{fig:weaktwocolornphard}
\end{figure}

An unbalanced unfriendly partition of $G$ corresponds to a two-coloring of $G$
in which the colors occur equally often. Partially coloring $G'$ in this way, we
can achieve stability by coloring $u,v$ opposite colors, coloring $w$ the same
color as $u$, and using this to stabilize the 3-cycle, as shown in
Figure~\ref{fig:weaktwocolornphard}. Conversely, suppose $G$ does not have a
balanced unfriendly partition and fix a stable 2-coloring of $G'$. WLOG suppose
$G$ has an even number of vertices and suppose color 1 occurs more often among
the vertices coming from $G$. Then $u,v$ must both have color 2, and hence $w$
has color 1. Since $u,w$ have different colors, the 3-cycle will not be stable.
This completes the reduction.

\noindent \emph{2)} $k\ge3$:
We reduce from the case of $k=2$. The idea is to augment the construction $G'$
above by disallowing all but two colors to be used in the $G'$ part. We call the
larger construction $G''$. 

We start with $G'' = G'$ add two new vertices $x,y$ to $G''$ which are adjacent
to each other. In a stable coloring, $x$ and $y$ will necessarily have different
colors (in our construction they will not be the tail of any other edges). We
call these colors 1 and 2, and will force them to be used in coloring $G'$.
Specifically, let $n$ be the number of vertices of $G'$, and construct $n^3$
copies of $K_{k-2}$. For each vertex $v$ in any copy of $K_{k-2}$, add the edges
$(v,x), (v,y)$. Finally, add all edges $(a,b)$ where $a \in G'$ and $b$ comes
from a copy of $K_{k-2}$. Figure~\ref{weakkcolornphard} shows this construction. 

\begin{figure}[htb]
\centering
\scalebox{0.4}{\includegraphics{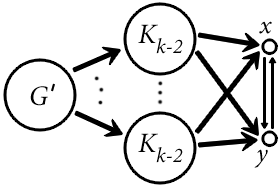}}
\caption{Reducing $k$ colors to two colors. A bold arrow indicates complete
incidence from the source subgraph to the target subgraph.}
\label{weakkcolornphard}
\end{figure}

Now in a stable coloring any vertex from a copy of $K_{k-2}$ must use a
different color than both $x,y$, and the vertex set of a copy of $K_{k-2}$ must
use all possible remaining $k-2$ colors. By being connected to $n^3$ copies of
$K_{k-2}$, each $a \in G'$ will have exactly $n^3$ neighbors of each of the
$k-2$ colors. Even if $a$ were connected to all other vertices in $G'$ and they
all use color 1, it is still better to use color 1 than to use any of the colors
in $\left \{ 3, \dots, k \right \}$. The same holds for color 2, and hence we
force the vertices of $G'$ to use only colors 1 and 2.  \hfill $\square$
\end{proof}

\section{Discussion and open problems}

In this paper we defined new notions of graph coloring.  Our results elucidated
anti-coordination behavior, and solved some open problems in related areas.

Many interesting questions remain.  For instance, one can consider alternative
payoff functions. For players choosing colors $i$ and $j$, the payoff $|i-j|$ is
related to the \emph{channel assignment problem}~\cite{vandenHeuvel98}.  In the
cases when the coloring problem is hard, as in our problem and the example
above, we can find classes of graphs in which it is feasible, or study random
graphs in which we conjecture colorings should be possible to find.  Another
variant is to study weighted graphs, perhaps with weights, as distances,
satisfying a Euclidian metric.

\subsubsection*{Acknowledgements}
We thank Gy\"{o}rgy Tur\'{a}n for helpful discussions.

\bibliographystyle{plain}
\bibliography{paper}

\begin{thebibliography}{10}

\bibitem{AharoniMP90}
R.~Aharoni, E.~C. Milner, and K.~Prikry.
\newblock Unfriendly partitions of a graph.
\newblock {\em J. Comb. Theory, Ser. B}, 50(1):1--10, 1990.

\bibitem{BazganTV10}
C.~Bazgan, Z.~Tuza, and D.~Vanderpooten.
\newblock Satisfactory graph partition, variants, and generalizations.
\newblock {\em Eur. J. Oper. Res.}, 206(2):271--280, 2010.

\bibitem{BramoulleLGV04}
Y.~Bramoull{\'e}, D.~L{\'o}pez-Pintado, S.~Goyal, and F.~Vega-Redondo.
\newblock Network formation and anti-coordination games.
\newblock {\em Int. J. Game Theory}, 33(1):1--19, 2004.

\bibitem{BruhnDGS10}
H.~Bruhn, R.~Diestel, A.~Georgakopoulos, and P.~Spr{\"u}ssel.
\newblock Every rayless graph has an unfriendly partition.
\newblock {\em Combinatorica}, 30(5):521--532, 2010.

\bibitem{CaoY12a}
Z.~Cao and X.~Yang.
\newblock The fashion game: Matching pennies on social networks.
\newblock {\em SSRN}, 2012.

\bibitem{CKPS10}
I.~Chatzigiannakis, C.~Koninis, P.~N. Panagopoulou, and P.~G. Spirakis.
\newblock Distributed game-theoretic vertex coloring.
\newblock In {\em OPODIS'10}, pages 103--118, 2010.

\bibitem{ChaudhuriGJ08}
K.~Chaudhuri, F.~C. Graham, and M.~Shoaib Jamall.
\newblock A network coloring game.
\newblock In {\em WINE}, pages 522--530, 2008.

\bibitem{CowanE}
R.~Cowen and W.~Emerson.
\newblock Proportional colorings of graphs.
\newblock unpublished.

\bibitem{ElsasserT11}
R.~Els{\"a}sser and T.~Tscheuschner.
\newblock Settling the complexity of local max-cut (almost) completely.
\newblock In {\em ICALP (1)}, pages 171--182, 2011.

\bibitem{EGM12}
B.~Escoffier, L.~Gourv\`{e}s, and J.~Monnot.
\newblock Strategic coloring of a graph.
\newblock In {\em CIAC'10}, pages 155--166, Berlin, Heidelberg, 2010.
  Springer-Verlag.

\bibitem{FKKMS02}
D.~Fotakis, S.~Kontogiannis, E.~Koutsoupias, M.~Mavronicolas, and P.~Spirakis.
\newblock The structure and complexity of nash equilibria for a selfish routing
  game.
\newblock In {\em ICALP}, pages 123--134, Malaga, Spain, 2002.

\bibitem{GareyJ79}
M.~R. Garey and D.~S. Johnson.
\newblock {\em Computers and Intractability: A Guide to the Theory of
  NP-Completeness}.
\newblock W. H. Freeman \& Co., New York, NY, USA, 1979.

\bibitem{G09}
L.~Gourv\`{e}s and J.~Monnot.
\newblock On strong equilibria in the max cut game.
\newblock In {\em In: Proc. of WINE 2009, Springer LNCS}, pages 608--615, 2009.

\bibitem{Hoefer2007}
M.~Hoefer.
\newblock {\em Cost sharing and clustering under distributed competition}.
\newblock PhD thesis, Universit{\"a}t Konstanz, Germany, 2007.

\bibitem{KearnsSM06}
M.~Kearns, S.~Suri, and N.~Montfort.
\newblock A behavioral study of the coloring problem on human subject networks.
\newblock {\em Science}, 313:2006, 2006.

\bibitem{KP99}
E.~Koutsoupias and C.~Papadimitriou.
\newblock Worst-case equilibria.
\newblock In {\em STACS}, pages 404--413, Trier, Germany, 4--6 March 1999.

\bibitem{M96}
Dov Monderer and Lloyd~S. Shapley.
\newblock Potential games.
\newblock {\em Games and Economic Behavior}, 14(1):124 -- 143, 1996.

\bibitem{MonienT10}
B.~Monien and T.~Tscheuschner.
\newblock On the power of nodes of degree four in the local max-cut problem.
\newblock In {\em CIAC}, pages 264--275, 2010.

\bibitem{NaorS93}
M.~Naor and L.~Stockmeyer.
\newblock What can be computed locally?
\newblock In {\em STOC '93}, pages 184--193. ACM, 1993.

\bibitem{PS08}
P.~N. Panagopoulou and P.~G. Spirakis.
\newblock A game theoretic approach for efficient graph coloring.
\newblock In {\em ISAAC '08}, pages 183--195, 2008.

\bibitem{RT02}
T.~Roughgarden and \'{E}. Tardos.
\newblock How bad is selfish routing?
\newblock {\em J. ACM}, 49(2):236--259, March 2002.

\bibitem{ShafiqueD09}
K.~Shafique and R.~D. Dutton.
\newblock Partitioning a graph into alliance free sets.
\newblock {\em Discrete Mathematics}, 309(10):3102--3105, 2009.

\bibitem{ShelahM90}
S.~Shelah and E.~C. Milner.
\newblock Graphs with no unfriendly partitions.
\newblock {\em A tribute to Paul Erd{\"o}s}, pages 373--384, 1990.

\bibitem{vandenHeuvel98}
J.~van~den Heuvel, R.~A. Leese, and M.~A. Shepherd.
\newblock Graph labeling and radio channel assignment.
\newblock {\em J. Graph Theory}, 29(4):263--283, December 1998.

\end{thebibliography}

\end{document}